\documentclass[12pt]{amsart}
\usepackage{jeffmath}
 \usepackage[normalem]{ulem}
\usepackage[margin=1in]{geometry}
\usepackage[colorlinks,linkcolor=blue,citecolor=blue]{hyperref}

\theoremstyle{plain}
\newtheorem{theorem}{\protect\theoremname}

\makeatletter
\@ifundefined{textcolor}{}
{%
 \definecolor{BLACK}{gray}{0}
 \definecolor{WHITE}{gray}{1}
 \definecolor{RED}{rgb}{1,0,0}
 \definecolor{GREEN}{rgb}{0,1,0}
 \definecolor{BLUE}{rgb}{0,0,1}
 \definecolor{CYAN}{cmyk}{1,0,0,0}
 \definecolor{MAGENTA}{cmyk}{0,1,0,0}
 \definecolor{YELLOW}{cmyk}{0,0,1,0}
}
\makeatother

\makeatletter
\newcommand{\setasstag}[1]{
  \let\oldtheass\theass
  \renewcommand{\theass}{#1}
  \g@addto@macro\endass{
    \addtocounter{ass}{-1}
    \global\let\theass\oldtheass}
  }

\numberwithin{equation}{section}

\providecommand{\theoremname}{Theorem}

\begin{document}

\title[An ergodic theorem for quantum processes]{An ergodic theorem for quantum processes with applications to matrix product states}

\author{Ramis Movassagh}
\email{ramis@us.ibm.com}
\address{IBM Quantum, MIT-IBM Watson AI Lab, Cambridge MA, 02142, USA}

\author{Jeffrey Schenker}
\email{jeffrey@math.msu.edu}
\address{Department of Mathematics, Michigan State University, East Lansing	MI, 48824, USA}

\begin{abstract}
Any discrete quantum process is represented by a sequence of quantum channels.  
We consider ergodic quantum processes obtained by a map that takes the points along the trajectory of a discrete ergodic dynamical system to the space of quantum channels.
Under a natural irreducibility condition, we obtain a theorem showing that the state under such a process converges exponentially fast to an ergodic sequence 
depending on the process, but independent of the initial state. As an application, we describe the thermodynamic limit of ergodic matrix product states and prove that the 2-point correlations of local observables in such states decay exponentially with their distance in the bulk.  
\end{abstract}

\maketitle

\section{Introduction}
The change of a physical system over a discrete unit of time, including the internal dynamics and interaction with the environment, can be represented by a quantum channel.
The evolution of the system at discrete times is then obtained by the application of a sequence of quantum channels, which may be termed a \emph{quantum process}. Mathematically, a quantum channel is a completely positive and trace
preserving linear transformation of the system's density matrix, $\rho \mapsto \phi(\rho)$. In a finite dimensional Hilbert space, any such map can be written in the Kraus form \cite{krausk1971}
\begin{equation}
\phi(\rho)=\sum_{i=1}^{d}B^{i}\;\rho\; B^{i\:\dagger}\quad,\label{eq:phi_channel}\end{equation}
where $\dagger$ denotes the adjoint (conjugate transpose) and the following holds
\begin{equation}
\sum^d_{i=1}B^{i\: \dagger }B^i \ = \bbI \quad .\label{eq:tracepreserving}
\end{equation}
The net change in the state resulting from the quantum process is obtained from composition of the channels acting on the initial state
\begin{equation}
\rho_n \ = \ \phi_{n}\circ\cdots\circ\phi_{1}(\rho_0)=\sum^d_{i_{1},\dots,i_{n}=1}B_{n}^{i_{n}}\cdots B_{1}^{i_{1}} \: \rho_0\: B_{1}^{i_{1}\: \dagger}\cdots B_{n}^{i_{n} \: \dagger}\quad.\label{eq:QChannel_Composition}
\end{equation}

In the present work, we study general \emph{ergodic sequences} of  channels in the following sense.  Consider a map from the points of a discrete, ergodic dynamical system $\Omega$ to the space quantum channels.  Starting from any point on $\Omega$, we obtain an ergodic sequence of quantum channels by evaluating the map at the points along the corresponding trajectory. 

Here we answer the following questions: What is the action of an ergodic composition given by  equation \eqref{eq:QChannel_Composition}? Is there a convergence to a simple and general limit? We obtain a general theorem (Theorem \ref{thm:main}) for an ergodic sequence of quantum channels, with an underlying assumption of non-negligible decoherence.  This theorem states that the sequence of states $\rho_n$ converges to a fixed-point sequence that only depends on the sequence of channels and is independent of the initial state. Theorem \ref{thm:bound} then shows that the composition of such channels converges exponentially fast to a stochastic sequence of replacement (rank-one) channels. A corollary of this result is the well-known convergence in the translation invariant case to a fixed replacement channel.

Theorems \ref{thm:main} and \ref{thm:bound} also apply to sequences of completely positive maps, without imposing the trace preserving condition equation \eqref{eq:tracepreserving}.  Such sequences are naturally related to the expectation values of observables in a matrix product state (MPS).  We apply our results to an ergodic MPS, wherein the matrices in the MPS form an ergodic sequence.  We derive a formula for the
expectation values of observables in an MPS. We then  prove (Theorem \ref{thm:corelations}) that the correlation functions of local observables decay exponentially
with their distance. 

\subsection{Background and relation to other works}The generic aspects of the behavior of quantum systems have long been of interest. However, because of the theoretical challenge of dealing with the general case, 
in the past `ergodic' quantum channels were considered in various works, each of which, to the best of our knowledge,
is a very special subset of possibilities
in this work.
For example, in \cite{burgarth2013ergodic}, a channel was chosen at random from some ensemble
and then repeatedly applied, i.e., the sets $\setb{B_{k}^i}{i=1,\ldots,d}$ were all equal. In \cite{bruneau2014repeated}, time dynamics were analyzed for a quantum system with repeated independently chosen random interactions with an environment. Other instances studied include certain independent random channels and their compositions (e.g., from a finite set of random isometries) \cite{collins2010random,collins2011random}. See \cite{collins2016random} for a review. Our work considers a general ergodic sequence and therefore serves as a vast generalization of the past work. In particular, this work allows for long-range correlations among the channels, or even pseudo-randomness generated by quasi-periodic dynamics. This includes the previously considered extreme cases of independently and identically distributed (iid) and (time)-translation invariant channels. 

The formalism of quantum channels naturally lends itself to the calculation of expectation values of observables and correlation functions of local observables of physical low-dimensional quantum systems, which are well described by density matrix renormalization  group \cite{white1992density} and its natural representation in terms of MPS \cite{verstraete2008matrix}. Previous works on matrix product states have focused on the translation invariant case \cite{fannes1991,perez2006matrix,brandao2013area}. Theorem \ref{thm:bound} allows us to move beyond the translation invariant case to analyze the thermodynamic limit of ergodic (one-dimensional) MPS. 
The ergodic MPS that we consider may be translation invariant, quasi-periodic, or formed from random matrices with arbitrary correlations.  

In order to reify our theory, in a companion paper \cite{movassagh2021} we apply our main result (Theorem \ref{thm:bound}) to the translation invariant case as well as a natural example in which each channel is an independent random Haar isometry:
$$\phi_j(\rho)  \ = \ \tr_r[U_j\rho\otimes Q_r U_j^\dagger] \quad ,$$
where $Q_r$ is a pure state on $\bb{C}^{r\times r}$, $U_j$ is a sequence of independent Haar distributed $Dr\times Dr$ unitaries, and $\tr_r$ is the partial trace from $\bb{C}^{Dr\times Dr}$ to $\bb{C}^{D\times D}$. We analyze the asymptotics with respect to the dimension of the environment ($r$), or of the system ($D$), or both tending to infinity, and prove that the limiting states $\rho_n$ are given by
$$ \rho_n \ = \ \frac{1}{D} \bb{I}_D + \frac{1}{\sqrt{1+rD^2}} W_n \quad , $$
where $W_n$ are asymptotically Gaussian with distribution proportional to $\e^{-\frac{D}{2} \tr[ W^2]} \delta(\tr [W]),$ where $\delta$ is the Dirac delta measure. We also present consequences for ergodic MPSs in \cite{movassagh2021}, using Theorem \ref{thm:corelations} and the theory developed in \S\ref{sec:emps} of the present paper to analytically compute the entanglement spectrum
of an ergodic MPS across any cut
as well as the one- and two-point correlation functions in an ergodic MPS. 

\subsection{Physical implications} 
In the companion paper \cite{movassagh2021}, we discuss the physical consequences of the theorems presented here.
Theorems  \ref{thm:main} and \ref{thm:bound}, to the best of our knowledge, are the first general theorems proved that apply to correlated quantum processes. Similarly, Theorem \ref{thm:corelations} for the first time demonstrates a general exponential decay of correlations for ergodic MPS, and therefore, a vast class of ground states of interacting quantum matter. 

Physically realistic quantum processes inevitably have temporal correlations, even if the underlying process is Markovian. In the latter case, any two consecutive times are correlated.  Similarly, correlated quantum channels arise naturally in the context of MPS for the study of non-trivial systems and states of interacting quantum many-body systems.  For example, in any finite system simulation of one-dimensional systems, the matrices that result in the density matrix renormalization group procedure will inevitably be correlated.  As such the consideration of iid channels and the MPS formed from them is mostly of theoretical interest.

Three physical corollaries of our theorems are~\cite{movassagh2021}: 

{\bf (1)} Engineered non-equilibrium phases of matter realized by time-periodic driven Hamiltonians (e.g., in Floquet systems) \cite{li2019observation,lindner2011floquet,shtanko2018stability}, are {\it only} meta-stable in presence of interactions with an environment at positive temperature.

{\bf (2)} An ergodic sequence of quantum channels with non-negligible decoherence converges to the same final sequence irrespective of the initial state. These channels are asymptotically replacement channels, which implies that the process cannot even convey classical information with respect to the initial state. This is intuitively seen from a unique fixed point that is reached irrespective of the input quantum state. The channels may be very correlated or even time-translation invariant however. For example, in the near-term quantum computing era when the random quantum circuits have decoherence at each step of computation, the initial memory of the state is exponentially lost with the number of applied gates.

{\bf (3)}  It was previously proved that a non-vanishing gap in the thermodynamic limit implies an exponential decay of correlations \cite{Nachtergaele2006,hastings2006spectral}, and that in one-dimension an area law for entanglement entropy holds \cite{hastings2007area}. Brand{\~a}o and Horodecki \cite{brandao2013area} proved that in one-dimension the exponential decay of correlations implies an area law. We prove somewhat of a partial converse, that says finitely correlated states with an ergodic MPS representation have correlation function that decay exponentially with distance.

\subsection{Illustration}
\begin{figure}
	\centering
	\includegraphics[width=\textwidth]{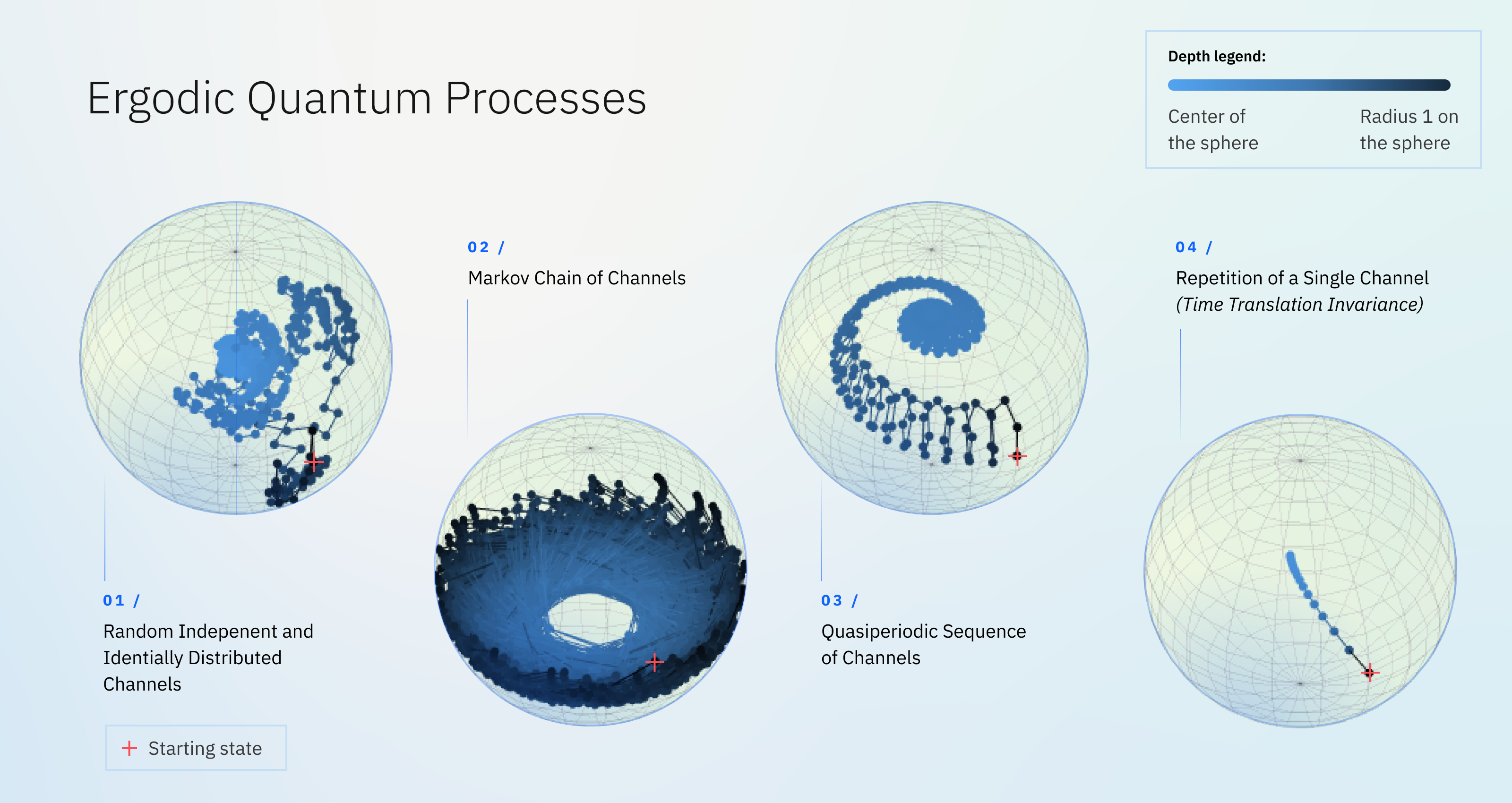}
	\caption{\label{fig:Bloch}Evolution of the state of a qubit under four ergodic quantum processes.}
\end{figure}
One of the simplest examples of a channel is the amplitude damping channel of a qubit.  The state space of a single qubit can be identified with the closed unit ball in $\mathbb{R}^3$, the so-called \emph{Bloch sphere}. A point $\vec{v}=(v_1,v_2,v_3)$ in the Bloch sphere corresponds to the density matrix $\rho_{\vec{v}}= \frac{1}{2}(I+\vec{v}\cdot \vec{\sigma})$, where
$\vec{\sigma}=(\sigma_1,\sigma_2,\sigma_3)$ are the Pauli matrices $$\sigma_1 = \left ( \begin{matrix}
0 & 1 \\ 1 & 0
\end{matrix} \right )
\ , \quad \sigma_2 = \left ( \begin{matrix}
	0 & -i \\ i & 0
\end{matrix} \right )\ ,\quad  \text{and} \quad  \sigma_3  =  \left ( \begin{matrix}
1 & 0 \\ 0 & -1
\end{matrix}\ \right ) .$$
The amplitude damping channel with rate $\gamma \in [0,1]$ and axis $\vec{u}=(\sin \theta \cos \phi, \sin \theta \sin \phi, \cos\theta)$ is the completely positive trace preserving map
$$\phi_{\gamma,\vec{u}} (\rho) \ = \  E\: \rho\: E^\dagger + F\: \rho\: F^\dagger \ , $$
where

\begin{equation*}E \ = \ S_{\vec{u}} \left (\begin{matrix} 1 & 0 \\ 0 & \sqrt{1-\gamma} \end{matrix} \right ) S_{\vec{u}}^\dagger \ \quad\text{and}\quad
F \ = \ S_{\vec{u}} \left (\begin{matrix} 0 & \sqrt{\gamma}  \\ 0  & 0 \end{matrix} \right ) S_{\vec{u}}^\dagger\;,
\end{equation*}
with $S_{\vec{u}}$ defined by
\begin{equation*}S_{\vec{u}}= \left  (\begin{matrix}
	\cos \nicefrac{\theta}{2} & -e^{-i \phi} \sin \nicefrac{\theta}{2} \\
	e^{i\phi} \sin \nicefrac{\theta}{2} & \cos \nicefrac{\theta}{2}
\end{matrix} \right ) \ . 
\end{equation*}
This channel models relaxation, with the rate $\gamma$, of a qubit to the pure state $\frac{1}{2}(I+\vec{u}\cdot \vec{\sigma})$, i.e., the ground state of the spin Hamiltonian $H_{\vec{u}}=-\vec{u}\cdot \vec{\sigma}$. More generally, one may consider relaxation coupled with the Schr\"odinger dynamics of $H_{\vec{u}}$ over an interval $\delta$ to obtain the generalized amplitude damping channel
$$\phi_{\gamma,\vec{u},\delta}(\rho) \ = \ \phi_{\gamma,\vec{u}} \left ( e^{-i\delta H_{\vec{u}}}\, \rho\, e^{i \delta H_{\vec{u}}} \right ) \ . $$

One may obtain a large family of ergodic quantum processes, for instance, by allowing $\vec{u}$ to evolve according to an ergodic process on the Bloch sphere $S^2$ (with $\gamma$ and $\delta$ fixed).  Such processes would model the relaxation of a qubit toward an axis that fluctuates with time, such as might be expected if the qubit Hamiltonian fluctuates while the interaction with the environment remains fixed.   The longtime behavior of such processes can depend quite strongly on the nature of the ergodic process on $S^2$.  More generally one could allow for all three parameters $(\gamma,\vec{u},\delta)$ to evolve according to an ergodic process on $[0,1]\times S^2 \times [0,\infty)$.

In Figure \ref{fig:Bloch} we plot the first 2000 steps $\rho_j=\phi_j\circ\cdots\phi_1(\rho)$, $j=1,\ldots,2000$ for four distinct ergodic quantum processes of the form $\phi_j=\phi_{\gamma,\vec{u}_j,\delta}$, $\vec{u}_j=(\cos \alpha_j,\sin \alpha_j,0)$. The four processes considered are
\begin{enumerate}
    \item Random channels, with $\delta=\pi/12$, $\gamma=0.01$ and $\alpha_j$ chosen independently and uniformly from $[0,2\pi)$.
    \item A Markov chain of channels with $\delta=\pi/12$, $\gamma=0.4$, $\alpha_1=0$ and $\alpha_{j+1}=\alpha_j- x_j$, where $x_j$ are independent and $x_j=1$ or $0$ each with probability one half.
    \item A quasi-periodic family of channels with $\delta=\pi/12$, $\gamma=0.01$, $\alpha_1=0$ and $\alpha_{j+1}=\alpha_j+1$.
    \item A periodic family of channels with $\delta=\pi/12$, $\gamma=0.1$, $\alpha_1=0$ and $\alpha_{j+1}=\alpha_j+2\pi/3$.  By sampling the sequence $\rho_j$ only at $j\equiv 0\: (\text{mod} 3)$, we get a trajectory obtained by repeating the single channel $\phi_3\circ \phi_2 \circ \phi_1$.
\end{enumerate}

\section{Ergodic theory of quantum processes and matrix product states}
\subsection{Notation}Let $\bbM_D=\bbC^{D\times D}$ denote the space of $D\times D$ matrices.  Recall the trace-norm of $M\in\bbM_D$, $ \norm{M}_1  =  \tr [|M|]$
as well as the Hilbert-Schmidt inner product and norm, $ \langle \wt{M}  , \, M \rangle =\tr[\wt{M}^{\dagger}M]$ and $\norm{M}_2^{2}=\tr[M^{\dagger}M]$.
Let $\mc{L}(\bbM_D)$ denote the set of linear maps from $\bbM_D$ to itself.
Given $\phi\in\mc{L}(\bbM_D)$, we define the adjoint map $\phi^{*}$ via the Hilbert-Schmidt inner product: $
\tr[\wt{M}^{\dagger}\phi(M)] = \tr[[\phi^{*}(\wt{M})]^{\dagger}M]$.
 
Let $\bb{P}_D$ denote the closed cone of positive semi-definite matrices in $\bbM_D$,
\[
\bb{P}_D=\setb{M\in \bb{M}_D}{ \mb{z}^\dagger M \mb{z}  \ge 0 \quad \text{for all }\mb{z}\in \bbC^D } \quad .
\]
The interior of $\bb{P}_D$
is the open cone of positive definite matrices,
\[
\bb{P}_D^\circ =\setb{M\in \bb{M}_D}{ \mb{z}^\dagger M \mb{z}  > 0 \quad \text{for all }\mb{z}\in \bbC^D \text{ with } \mb{z}\neq 0} \quad .
\]
A map $\phi \in \mc{L}(\bbM_D)$ is \emph{positive} if $\phi(\bb{P}_D) \subset \bb{P}_D$, i.e., $\phi$ maps positive semi-definite matrices to positive semi-definite matrices. The map is \emph{strictly positive} if $\phi(\bb{P}_D\setminus\{0\})\subset \bb{P}_D^\circ$, i.e., $\phi$ maps positive semi-definite matrices to positive definite matrices.  A \emph{completely positive} map is one such that $\phi\otimes \bbI_{r}:\mc{L}(\bbM_D)\otimes \mc{L}(\bbM_r ) $  is positive for every $r$, where $\bbI_{r}$ denotes the identity map on $\bbM_r$; let $\mc{CP}(\bbM_D)$ denote the set of completely positive maps over $\bbM_D$,
\[\mc{CP}(\bbM_D) \ = \ \setb{\phi\in \mc{L}(\bbM_D)}{\phi \text{ is completely positive.}} \quad . \]
By Kraus's theorem \cite{krausk1971,watrous2018theory}, $\phi\in \mc{CP}(\bbM_D)$ if and only if $\phi$ is of the form equation \eqref{eq:phi_channel}. 
A map $\phi \in \mc{L}(\bbM_D)$ is \emph{trace preserving} if $\tr [\phi(M)] =\tr [M]$ for all $M$; equivalently $\phi^*(\bbI_D)=\bbI_D$.
A \emph{quantum channel} is a completely positive trace preserving map.

Let $(\Omega,\mathcal{F},\Pr)$ be a probability space with 
\begin{enumerate}
    \item $T:\Omega\rightarrow\Omega$
an  invertible, ergodic, and measure preserving map, and
\item $\phi_0:\Omega \rightarrow \mc{CP}(\bbM_D)$ a completely positive map valued random variable (taking the Borel $\sigma$-algebra on $\mc{CP}(\bbM_D)$).
\end{enumerate}
 Recall that $T$ is \emph{ergodic} provided $\Pr[A]=0$ or $1$ for any measurable set $A$ with $T^{-1}(A)=A$. We follow the convention in probability theory and suppress the independent variable $\omega \in \Omega$ in most formulas; when it is needed we will use a subscript to denote the value of a random variable at a particular $\omega \in \Omega$, \textit{e.g.},  $\phi_{0;\omega}$. To specify $\phi_0$ we could introduce matrix-valued random variables $B^{i}_0:\Omega\rightarrow\bbM_D$,
for $i=1,\ldots,d$, and take
\begin{equation}\label{eq:phidefn}
\phi_0(M) \ = \ \sum_{i=1}^d B^{i}_0\; M \;  B_0^{i\:\dagger}\quad.
\end{equation}
If we further impose the
condition 
\begin{equation}\label{eq:channelcondition}
\sum_{i=1}^d  B_0^{i\:\dagger} B_0^{i} \ = \ \bbI_D \quad \text{ almost surely },
\end{equation}
then $\phi_0$ is almost surely trace preserving, so $\phi_0$ is almost surely equal to a quantum channel valued random variable.  We note, however, that the matrices $B_0^i$, $i=1,\ldots,d$, are not uniquely determined by the channel $\phi_0$.  For this reason, we formulate our results directly in terms of the channel valued random variable $\phi_0$ without reference to a specific Kraus representation.  

\subsection{Ergodic theorems for quantum processes}
The main focus of this paper is to study the composition of a sequence of maps obtained by evaluating $\phi_0$ along the trajectories of the ergodic map $T$:
\begin{equation}\label{eq:phin}
\phi_{n;\omega} \ = \ \phi_{0;T^n\omega}\quad ,
\end{equation}
with $n\in \bbZ$. 
For our general result, we do not require the maps to be quantum channels, i.e., trace preserving.  Nonetheless, we take inspiration from the quantum channel case and consider the dynamics $\rho_n = \phi_{n}(\rho_{n-1})$ induced by the sequence $(\phi_n)_{n=0}^\infty$ on (non-normalized) states of a $D$-dimensional quantum system with Hilbert space $\mc{H}=\bb{C}^D$.

Consider the  process 
\begin{equation}\label{eq:PhiDefn} \Phi_N \ = \ \phi_N \circ \cdots \circ \phi_0 
\end{equation}
for $N\ge 0$.  The only assumption we need is that
\begin{ass}\label{ass:main}
With probability one there exists an $N_0>0$ such that $\Phi_N$ is strictly positive for all $N\ge N_0$.
\end{ass}
\noindent Physically, this assumption states that no proper subspace of the system is invariant under the dynamics. For more discussion of the physical motivation behind Assumption \ref{ass:main}, see \cite{movassagh2021}

Although Assumption \ref{ass:main} is physically natural, it is not formulated in a way that is easily verifiable.  However, it is \emph{equivalent} to  two more easily verified assumptions:
\begin{lem}\label{lem:Assequiv} Assumption \ref{ass:main} is equivalent to the following two statements taken together:
\begin{enumerate}
    \item  For some $n_{0}>0$, $ 	\Pr\left[\Phi_{n_0}\text{ is strictly positive }\right]\ > \ 0 $.
    \item  With probability one,  $(\ker\phi_{0})\cap\bb{P}_D \ = \  (\ker\phi_{0}^{*})\cap\bb{P}_D \ =\ \{0\}$. \\
    That is, if $\phi_{0}(M)=0$ or $\phi_{0}^{*}(M)=0$ with $M\in\bb{P}_D$, then $M=0$.
\end{enumerate}
\end{lem}
\begin{rems*} 1) Conditions (1) and (2) are manifestly verifiable by a finite computation, while Assumption \ref{ass:main}, as stated, is not. The proof of Lemma \ref{lem:Assequiv} is given below in \S\ref{sec:existence}. 2) A map $\phi$ is strictly positive if and only if $\phi^*$ is strictly positive.\footnote{Indeed, if $\phi$ is strictly positive and $M\in \bb{P}_D$ is non-zero, then we have $\tr [\phi^*(M) M'] = \tr [M \phi(M')] >0$ for any non-zero $M'\in \bb{P}_D$, since $\phi(M')>0$.  Thus $\phi^*(M)$ is strictly positive.}  Thus condition (1)  is equivalent to $\phi_{0}^*\circ \cdots \circ \phi_{N_0}^*$ being strictly positive with positive probability.  3) If $\phi_0$ is trace preserving, i.e., a quantum channel, then $\tr [\phi_0(M)]=\tr[M]$ for any $M$, so $\ker \phi_0\cap \bb{P}_D=0$. However, the other half of condition (2) (that $\ker \phi_0^* \cap \bb{P}_D=\{0\}$) does not necessarily hold.  For example, if $D$ is even and $\phi(M)=PMP+SMS^\dagger$ with $P$ a projection onto a subspace of dimension $D/2$ and $S$ a partial isometry from $\mathbb{I}_D-P$ to $P$, then $\phi$ is a channel but $\phi^*(\mathbb{I}_D-P)=0$.
 \end{rems*}

The classical Perron-Frobenius theorem \cite{perron1907theorie,frobenius1912matrizen} has been generalized to linear maps preserving a convex cone, e.g., see \cite{kreinrutman}.  Based on such a generalization, Evans and H{\o}egh-Krohn \cite{evans1978spectral} obtained results for positive maps on $\bbM_D$. It follows from \cite[Theorem 2.3]{evans1978spectral} that,
if $\Phi_{N}$ is strictly positive, then there is a unique (up to
scaling) strictly positive matrix $R_{N}\in \bbM_D$ such that
\begin{equation}\label{eq:Phi_N}
\Phi_{N}(R_{N})=\lambda_{N}R_{N}\quad,
\end{equation}
where $\lambda_{N}$ is the spectral radius of $\Phi_{N}$.  Similarly,
there is a unique (up to scaling) strictly positive matrix $L_{N}$
such that
\begin{equation}\label{eq:PhiStar_N}
\Phi_{N}^{*}(L_{N})=\lambda_{N}L_{N}\quad.
\end{equation}

We extend the process to $-N<0$ by defining
\begin{equation}\label{eq:PhiDefnNeg}
\Phi_{-N} \ = \ \phi_0 \circ \cdots \circ \phi_{-N}  \ .
\end{equation}
By Assumption \ref{ass:main}, $\Phi_N$ is strictly positive for all sufficiently large $N>0$.  In Lemma \ref{lem:stopping} below, we show below that, with probability one, we also have $\Phi_{-N}$ strictly positive for all $N$ sufficiently large.  Thus the left and right eigen-matrices $R_N$ and $L_N$ are unique for large $|N|$.
We normalize $R_{N}$ and $L_{N}$ so that $ \tr [R_{N}]=\tr [L_{N}]=1$.

Our first result is that $L_N$ converges as $N\rightarrow \infty$, while $R_N$ converges as $N\rightarrow -\infty$.  
\begin{theorem}\label{thm:main} There are random
matrices $Z_{0},Z_{0}':\Omega\rightarrow \bbM_D$ such that $Z_{0},Z_{0}'\in \bb{P}_D^\circ$,
\[
\lim_{N\rightarrow-\infty}R_{N}=Z_{0} \quad , \qquad \text{and} \qquad 
\lim_{N\rightarrow\infty}L_{N}=Z_{0}'
\]
almost surely. Furthermore, if we set $Z_{n}=Z_{0;T^{n}\omega}$ and
$Z'_{n}=Z'_{0;T^{n}\omega}$, then
\[
Z_{n}=\phi_{n}\cdot Z_{n-1}\quad, \qquad \text{and} \qquad
Z_{n}'=\phi_{n}^{*}\cdot Z_{n+1}'\quad,
\]
where $\cdot$ denotes the projective action of a positive map on the
strictly positive $D\times D$ matrices of trace $1$:
\[
\phi_{n}\cdot M \ \equiv \ \frac{1}{\tr[\phi_{n}(M)]}\phi_{n}(M)\quad.
\]
\end{theorem} 
\begin{rems*} 1) If the maps $\phi_n$ are quantum channels, then $L_N=\frac{1}{D}\bbI$, so $Z_n'=\frac{1}{D}\bbI$ for all $n$.  2) This result is closely related in spirit to Oseledec's Multiplicative Ergodic Theorem \cite{oseledec1968multiplicative}, a  general result on convergence of singular vectors for products of linear transformations.   3) Theorem \ref{thm:main} generalizes a theorem of Hennion on the Perron-Frobenius eigenvectors of products of entry-wise positive matrices \cite{hennion1997limit}. In fact, Hennion's theorem can be seen as a special case of our result applied to the following maps
	\begin{equation}\label{eq:hennionform}
		\phi_0(M) \ = \ \sum_{\alpha,\beta} A_{0;\alpha,\beta} \mb{e}_\alpha \mb{e}_\beta^T M \mb{e}_\beta \mb{e}_\alpha^T \quad , 
	\end{equation}
	with $A_0$ a random matrix with non-negative entries, and $\mb{e}_\alpha$, $\alpha=1,\ldots,D$, the standard basis vectors of $\bbC^D$. An equivalent, simpler, expression to equation \eqref{eq:hennionform} is given by  $\phi_0(M)  =  \operatorname{diag}(A_0 \operatorname{vec}(M))$, where $\operatorname{vec}{M}$ is the $D$-dimensional vector consisting of the diagonal entries of $M$ and $\operatorname{diag}(\mb{v})$ is a diagonal matrix with the entries of the vector $\mb{v}$ on the diagonal.
\end{rems*}

Given $m<n$ in $\bbZ$, let $P_{n,m}$ denote the rank-one
operator
\begin{equation}
P_{n,m}(M)=\tr [Z'_{m}M ]\;Z_{n}\quad.
\end{equation}
Our second result states that, for $n-m$ large, the operator $\phi_{n}\circ\cdots\circ\phi_{m}$
is well approximated by $P_{n,m}$. To formulate this result precisely,
we use the operator norm for a map $\Phi\in \mc{L}(\bbM_D)$ inherited from the trace norm on $\bbM_D$, $
\norm{\Phi}_{1} \ =\ \max \setb{\tr[\:\abs{\Phi(M)}\:]\: } {\: \tr[\:|M|\:]=1} $.
\begin{theorem}\label{thm:bound}
Given $m<n$ in $\Z$, let $\Psi_{n,m}=\phi_{n}\circ\cdots\circ\phi_{m}$.
There is $0<\mu<1$ so that for each $x\in\bbZ$ the following bound holds:
\begin{equation}
\norm{\frac{1}{\tr[\Psi_{n,m}^{*}(\bbI)]}\Psi_{n,m}-P_{n,m}}_{1}\ \le\ C_{\mu,x}\, \mu^{n-m}\label{eq:main}
\end{equation}
for all $m\le x$ and $n\ge x$, with $C_{\mu,x}$
finite almost surely. \end{theorem}

\subsection{Ergodic matrix product states}\label{sec:emps}
The Kraus matrices associated to a random completely positive map specified as in equation \eqref{eq:phi_channel}
can be used to define a family of random matrix product states as follows. Let
\[
A_{n}^{i}=B_{0;T^{n}\omega}^{i\: \dagger }\quad,\qquad n\in\mathbb{Z}\text{ and }i=1,\ldots,d.
\]
Given an interval $[m,n]$ of $\mathbb{Z}$, we define the matrix product state
\begin{equation}
|\psi([m,n])\rangle=\frac{1}{\mathcal{N}}\sum_{i_{m},\ldots,i_{n}=1}^{d}\text{tr}[A_{m}^{i_{m}}\cdots A_{n}^{i_{n}} ] \ |i_{m},\ldots,i_{n}\rangle \quad , \label{eq:MPS-1}
\end{equation}
where $|i_m,\ldots,i_n\rangle$ are the elements of the computational basis on  $\bigotimes_{k=m}^{n}\mathbb{C}^{d}$, where $d=2$ corresponds to qubits, and the normalization constant is given by
\[
\mathcal{N}^{2}=\sum_{i_{m},\ldots,i_{n}=1}^{d}\left|\text{tr}[A_{m}^{i_{m}}\cdots A_{n}^{i_{n}}]\right|^{2}.
\]
For simplicity, we restrict our attention here to the periodic boundary condition states, as defined in equation \eqref{eq:MPS-1}. 

There is a close relation between matrix product states and completely positive maps, via which Theorem \ref{thm:bound} can be used to characterize the thermodynamic limit ($m\rightarrow-\infty$ and $n\rightarrow\infty$) of the states 
defined in   equation \eqref{eq:MPS-1}. A preliminary observation is that the normalization factor $\mathcal{N}$ can be  expressed as
\begin{equation}\label{eq:N2one}	
\mathcal{N}^{2}=\sum_{\alpha,\beta=1}^{D}\text{tr}\left[\mb{e}_\beta \mb{e}_\alpha^T \phi_{n}\circ\cdots\circ\phi_{m} \left(\mb{e}_\alpha \mb{e}_\beta^T\right)\right].
\end{equation}
Since $\setb{\mb{e}_\alpha \mb{e}_\beta^T}{\alpha,\beta=1,\ldots,D}$ is an orthonormal basis for $\bbM_D$, equation \eqref{eq:N2one} can, in turn, be written as
\[
\mathcal{N}^{2}=\text{Tr}[\phi_{n}\circ\cdots\circ\phi_{m}]\quad,
\]
where $\text{Tr}[\phi]$ denotes the trace of a linear map $\phi \in \mc{L}(\bbM_D)$.   Throughout this discussion, we use $\tr$ to denote the trace on $\bbM_D$ and $\Tr$ to denote the trace on $\mc{L}(\bbM_D).$

Let us now consider the state $|\psi(N)\rangle\equiv|\psi([-N,N])\rangle$ defined on $[-N,N]$ by equation \eqref{eq:MPS-1}.  Given $-N<m<n<N$ and a local  observable $O$ on the spins in $[m,n]$, let
\begin{equation}\label{eq:Ohatagain}
\widehat{O}(M)=\sum_{\substack{i_{m},\ldots,i_{n}\\
j_{m},\ldots,j_{n}}}
\langle i_{m},\ldots,i_{n}|O|j_{m},\ldots,j_{n}\rangle A_{n}^{i_{n}\: \dagger }\cdots A_{m}^{i_{m}\: \dagger }\;M \; A_{m}^{i_{m}}\cdots A_{n}^{i_{n}}
\end{equation}
which is a linear operator on $\bbM_D$. One may easily verify that the
(quantum) expectation of $O$ in $|\psi(N)\rangle$ 
is
\begin{equation}\label{eq:MPSformula}
\langle\psi(N)|O|\psi(N)\rangle =\frac{\text{Tr}\left[\phi_{N}\circ\cdots\circ\phi_{n+1}\circ\widehat{O}\circ\phi_{m-1}\circ\cdots\circ\phi_{-N}\right]}{\text{Tr}\left[\phi_{N}\circ\cdots\circ\phi_{-N}\right]}\quad .
\end{equation}

We can express the thermodynamic limit of $\langle\psi(N)|O|\psi(N)\rangle$ in terms of the matrices $Z_{m}$ and $Z_{m}'$ from Theorem \ref{thm:main} by using equation \eqref{eq:MPSformula} and Theorem \ref{thm:bound}. It is convenient to use Dirac notation for the Hilbert-Schmidt inner product on $\bbM_D$, with which 
we have $ P_{n,m}\ =\ \ket{Z_{n}}\bra{Z'_{m}}$ (with $P_{n,m}$ as in Theorem \ref{thm:bound}).
Let
\begin{equation}
W(O)=\lim_{N\rightarrow\infty}\dirac{\psi(N)}{O}{\psi(N)}\label{eq:TDlimit}
\end{equation}
denote the thermodynamic limit of $\ket{\psi(N)}$, where $O$
is any local observable. Using equation \eqref{eq:MPSformula} and Theorem \ref{thm:bound}, we 
compute $W(O)$ as follows:
\begin{equation}
 \label{eq:Wformula}
 \begin{aligned}
  W(O) \ &= \ \lim_{N\rightarrow\infty}\frac{\Tr\ket{Z_{N}}\dirac{Z'_{n+1}}{\wh{O}}{Z_{m-1}}\bra{Z'_{-N}}}{\Tr\ket{Z_{N}}\dirac{Z'_{n+1}}{\phi_{n}\circ\cdots\circ\phi_{m}}{Z_{m-1}}\bra{Z'_{-N}}} \\ &= \ \frac{\dirac{Z'_{n+1}}{\wh{O}}{Z_{m-1}}}{\dirac{Z'_{n+1}}{\phi_{n}\circ\cdots\circ\phi_{m}}{Z_{m-1}}}\quad ,
\end{aligned}
\end{equation}
whenever $O$ is a local observable on the spins in $[m,n]$. Since $Z_{m}=\phi_{m}\cdot Z_{m-1}$,
the normalization in the denominator is given by
\[
\dirac{Z'_{n+1}}{\phi_{n}\circ\cdots\circ\phi_{m}}{Z_{m-1}}\ =\ \left[\prod_{k=m}^{n}\tr[\, \phi_{k}(Z_{k-1}) \, ]\right]\diracip{Z'_{n+1}}{Z_{n}}\quad.
\]

As is well known, there is a ``gauge-freedom'' in the representation of a matrix product state such as $|\psi(N)\rangle$:  the state itself does not change under the replacement $A^i_k \mapsto V_{k-1}A^i_k V_k^{-1}$ provided we have $V_N=V_{-N}$. See \cite{perez2006matrix} for more discussion on this symmetry.  Choosing the matrices $V_k$ appropriately, one can impose a gauge fixing condition such as $ \sum_{i=1}^d A_i A_i^\dagger = \bbI,$
which would make the associated channel maps trace preserving.  \emph{A priori}, it appears that the matrices required for gauge fixing may depend on $N$ and it is not clear that they can be chosen consistently with the ergodic structure.  However, Theorem \ref{thm:main} allows us to do just that, as we now explain. To begin, let
\[
\xi_{m}=\tr [\phi_{m}^{*}\left(Z_{m+1}'\right)]\quad.
\]
By Theorem \ref{thm:main}, $(\xi_{m})_{m\in \bb{Z}}$  is a shift-covariant sequence ($\xi_{m;\omega}=\xi_{m-1;T\omega}$) of positive random variables.
Furthermore $
\phi_{m}^{*}(Z_{m+1}')=\xi_{m}Z_{m}'$.
Now let
\begin{equation}\label{eq:transformphi}
	\wt{\phi}_{m}(M)\ =\ \frac{1}{\xi_{m}}Z'^{\;1/2}_{m+1}\phi_{m}\left(Z'^{\;-1/2}_m \; M\;Z_{m}'^{-1/2}\right)Z_{m+1}'^{\;1/2}\quad;
\end{equation}
this expression is well defined since the matrices $Z_m'$ are full rank almost surely by Theorem \ref{thm:main}.  The maps $\wt{\phi}_m$ are an ergodic sequence of completely positive maps, and a short computation shows that they are trace preserving:
\begin{eqnarray*}
\tr [\wt{\phi}_{m}(M)] 
&=& \frac{1}{\xi_{m}}\tr\left[ Z_{m+1}'\phi_{m}\left(Z'^{\;-1/2}_{m}M\;Z'^{\;-1/2}_{m}\right) \right]= \frac{1}{\xi_{m}}\tr\left[ \phi_{m}^{*}(Z'_{m+1})\ Z'^{\;-1/2}_{m}\;M\;Z'^{\;-1/2}_{m}\right]\\
&=& \tr \left[Z_{m}'\;Z'^{\;-1/2}_{m}\;M\;Z'^{\;-1/2}_{m}\right] = \tr [M]\quad .
\end{eqnarray*}
Also,
\[
\wt{\phi}_{m}(M)\ =\ \sum_{i=1}^{d}{\wt{A}_{m}^{i\, \dagger} }\;M\;{\wt{A}_{m}^{i}}
\]
where $\wt{A}_{m}^{i}\ =\ \frac{1}{\sqrt{\xi_{m}}}Z'^{\;1/2}_{m+1}\;A_{m}^{i}\;Z'^{\;-1/2}_{m}$.

We could apply Theorem \ref{thm:main} directly to the sequence $\wt{\phi}_m$, since it is straightforward to see that these maps satisfy
conditions (1) and (2) of Lemma \ref{lem:Assequiv}.  
However, it is easier to simply write down the left and right matrices  $\wt{Z'}_{m}$ and $\wt{Z}_{m}$ directly using equation \eqref{eq:transformphi}:
\[
\wt{Z'}_{m}=\frac{1}{D}\mathbb{I}\quad\text{and}\quad\wt{Z}_{m}=\frac{1}{\tr [Z_{m+1}'Z_{m}]} \ Z'^{\;1/2}_{m+1}\;Z_{m}\;Z'^{\;1/2}_{m-1}\quad. \]
Since $\wt{\phi}_{m}$ is trace preserving we have $\wt{\phi}_{m}(\wt{Z}_{m-1})=\wt{\phi}_{m}\cdot\wt{Z}_{m-1}=\wt{Z}_{m}$ 
for all $m$ and $
\wt{\phi}_{m}^{*}(\mathbb{I})=\mathbb{I}$.

We now return to the expression for the thermodynamic limit $W(O)$.  Given an interval $[m,n]$, one may easily check that
\begin{multline*}
\dirac{Z'_{n+1}}{\phi_{n}\circ\cdots\circ\phi_{m}}{Z_{m-1}} \\ = \ (\xi_{n}\cdots\xi_{m})\tr[ Z_{m}'Z_{m-1}]\tr[\wt{\phi_{n}}\circ\cdots\circ\wt{\phi}_{m}(\wt{Z}_{m-1})]\ = \ (\xi_{n}\cdots\xi_{m})\tr [Z_{m}'Z_{m-1}] \quad .
\end{multline*}
For a local observable $O$ on the spins in $[m,n]$, we define analogous to   equation \eqref{eq:Ohatagain},
\begin{equation}\label{eq:Otilde}
\wt{O} (M)\ \equiv\ \sum_{\substack{i_{m},\ldots,i_{n}\\
j_{m},\ldots,j_{n}
}
}\dirac{i_{m},\ldots,i_{n}}{O}{j_{m},\ldots,j_{n}}\wt{A}_{n}^{i_{n}\: \dagger }\cdots\wt{A}_{m}^{i_{m}\: \dagger }\; M\; \wt{A}_{m}^{i_{m}}\cdots\wt{A}_{n}^{i_{n}} \quad .
\end{equation}
Inserting these definitions into   equation \eqref{eq:Wformula}, we find the following remarkably simple formula for the thermodynamic limit $W$ of the matrix product states:
\begin{equation}
\label{eq:Wformulatilde}
W(O)\ =\ \tr\left[\wt{O}(\wt{Z}_{m-1})\right] \quad .
\end{equation}

  Equation \eqref{eq:Wformulatilde} can be used to obtain a bound on the two-point correlation of two observables $O_1$ and $O_2$ located in disjoint intervals $I_{1}=[m_{1},n_{1}]$ and $I_{2}=[m_{2},n_{2}]$ with $n_1 < m_2$. For such observables
\[
W(O_{2}O_{1})\ =\ \tr\left[\wt{O}_{2}\circ\wt{\phi}_{m_{2}-1}\circ\cdots\circ\wt{\phi}_{n_{1}+1}\circ\wt{O}_{1}(\wt{Z}_{m_{1}-1})\right] \quad .
\]
Applying Theorem \ref{thm:bound} to $\wt{\phi}_{m}$ allows us to obtain the following
\begin{theorem}\label{thm:corelations} There is $0<\mu<1$ such
that for each $x\in\bbZ$ the following correlation inequality holds with $C_{\mu,x}<\infty$ almost surely:
\begin{equation}
\abs{W(O_{2}O_{1})-W(O_{2})W(O_{1})}\ \le\ C_{\mu,x}\norm{\wt{O}_{1}-W(O_{1})\wt{\Psi}_{1}}\norm{\wt{O}_{2}-W(O_{2})\wt{\Psi}_{2}}\mu^{m_{2}-n_{1}},\label{eq:corelations}
\end{equation}
whenever  $\supp[{O_{j}}]\in[m_j,n_j]$ and $\wt{\Psi}_{j}=\wt{\phi}_{n_{j}}\circ\cdots\circ \wt{\phi}_{m_{j}}$ for $j=1,2$ with
$n_{1}<x<m_{2}$. 
\end{theorem}

\subsection{Overview of the proofs} 
Theorems \ref{thm:main}, \ref{thm:bound}, and \ref{thm:corelations} are proved in \S\ref{sec:proofs} below, using several technical lemmas presented and proved in \S\ref{sec:technical}. 
The central idea of the proofs 
is contraction mapping argument for the maps $\Phi_N$ on $\bb{P}_D$.  This is accomplished in Lemma \ref{lem:contract}, facilitated by the introduction of a non-standard metric on the set of quantum states \textemdash \ see equation \eqref{eq:d} below.  The metric and a number of the ideas developed in \S\ref{sec:technical} are inspired by results in section 10 of Hennion's paper \cite{hennion1997limit}.  Although some of the statements are similar, the proofs in \cite{hennion1997limit} do not directly carry over to the present more general context. Nonetheless, for readers interested in comparing the two papers, we note the following correspondence between lemmas in the present paper and in \cite{hennion1997limit}:
\[ \text{Lemma }
\left \{ \begin{matrix}\text{\ref{lem:m}} \\ 
	\text{\ref{lem:isametric}} \\
	\text{\ref{lem:dformula}} \\
	\text{\ref{lem:metric}} \\
	\text{\ref{lem:homeo}}\\
	\text{\ref{lem:contract}}
	\end{matrix} \right \} \text{ generalizes Lemma }  \left \{ \begin{matrix}\text{10.1} \\ 
	\text{10.2} \\
	\text{10.3} \\
	\text{10.4} \\
	\text{10.5}\\
	\text{10.6}
\end{matrix} \right \} \text{ of \cite{hennion1997limit}.}
\]
Although we have taken inspiration from \cite{hennion1997limit}, we do not use any of the results therein directly and the present paper can be read on its own.

\section{Technical results}\label{sec:technical}
\subsection{Notation}\label{sec:notation}
Let
\[
\bb{S}_D=\setb{M\in\bb{P}_D}{\tr[M]=1}\ ,\quad 
\bb{S}_D^\circ=\setb{M\in \bb{P}_D^\circ}{\tr[M]=1} \ ,
\]
and 
$$\mc{P}_D=\setb{\phi \in \mc{L}(\bbM_D)}{\phi \text{ is a positive map, }\ker\phi\cap\bb{P}_D=\{0\}\text{, and }\ker\phi^{*}\cap\bb{P}_D=\{0\}.}.$$
Note that $\mc{P}_D$ is a convex set. Let $\mc{P}_D^\circ$ denote its interior, 
Since any strictly positive map satisfies the kernel condition in the definition of $\mc{P}_D$, we have
$$\mc{P}_D^{\circ}=\setb{\phi \in \mc{L}(\bbM_D)}{\phi \text{ is strictly positive.}}.$$
Assumption \ref{ass:main} ensures that $\Phi_N=\phi_{N}\circ\cdots \circ \phi_0 \in \mc{P}_D^\circ$ for large $N$, with probability one. 
Condition (2) of Lemma \ref{lem:Assequiv} states that $\phi_{0}\in \mc{P}_D$ almost
surely, while
condition (1)
states that $\Phi_{n_0}\in \mc{P}_D^\circ$
with positive probability for some positive integer  $n_0$.
Note that any $\phi\in \mc{P}_D$ maps $\bb{P}_D$ into $\bb{P}_D$, while $\phi\in \mc{P}^{\circ}_D$ maps $\bb{P}_D$ into $\bb{P}_D^\circ$.

Theorem \ref{thm:main} is formulated in terms of the projective action
\[
\phi\cdot M\ \equiv\ \frac{\phi(M)}{\tr[\phi(M)]}\quad,
\]
of a positive map on $\bb{S}_D$. Note that $\tr[\phi(M)]\neq0$
for $\phi\in \mc{P}_D$ and $M\in\bb{S}_D$, so this action is well defined.

\begin{lem}\label{lem:phiC}Let $\phi\in \mc{P}_D$ then $\phi$ maps $\bb{P}_D^\circ$
into $\bb{P}_D^\circ$. \end{lem} \begin{proof} We first show that $\phi(\mathbb{I})\in \bb{P}_D^\circ$.
Suppose on the contrary that $\phi(\mathbb{I})\in\bb{P}_D\setminus \bb{P}_D^\circ$.
Let $P$ denote the orthogonal projection onto the kernel of $\phi(\mathbb{I})$.
Then $
0=\tr [P\phi(\mathbb{I})]=\ \tr[\phi^{*}(P)]$, 
so $\phi^{*}(P)=0$, contradicting the definition of $\mc{P}_D$. Thus $\phi(\mathbb{I})\in \bb{P}_D^\circ$. Now let $M$ be any point of $\bb{P}_D^\circ$ and let $\delta>0$ such that $M\ge\delta \mathbb{I}$.
Then $\phi(M)\ge\delta\phi(\mathbb{I})$, so $\phi(M)\in \bb{P}_D^\circ.$ \end{proof}

The sets $\mc{P}_D,\mc{P}_D^{\circ}$ are semi-groups under composition; it follows
from Lemma \ref{lem:phiC} that $\mc{P}_D^{\circ}$ is a two-sided ideal of $\mc{P}_D:$
\begin{cor}\label{cor:ideal} Given $\phi\in \mc{P}_D$ and $\phi'\in \mc{P}_D^{\circ}$,
we have $\phi\circ\phi'\in \mc{P}_D^{\circ}$ and $\phi'\circ\phi\in \mc{P}_D^{\circ}$.  \end{cor} \begin{proof}
We have $\phi'\circ\phi$  in $\mc{P}_D^{\circ}$, since $\phi'\circ\phi(\bb{P}_D)\subset\phi'(\bb{P}_D) \subset \bb{P}_D^\circ.$
On the other hand, for any $M\in\bb{P}_D$ we have $\phi'(M)\in \bb{P}_D^\circ$
and thus $\phi\circ\phi'(M)\in \bb{P}_D^\circ$, by Lemma \ref{lem:phiC}.\end{proof} 

\subsection{Geometry of $\bb{S}_D$} 
The set $\bb{S}_D$ of density matrices is convex and compact. To implement the contraction argument at the heart of the proof of Theorem \ref{thm:main}, it is useful to introduce a special metric on this space based on the following quantity:
\begin{equation}\label{eq:m}
m(X,Y)=\sup\setb{\lambda}{\lambda Y\le X} \quad ,
\end{equation}
for $X,Y\in \bb{S}_D$.
\begin{lem}\label{lem:m} Let $X,Y \in \bbS_D$.  Then
\begin{equation}\label{eq:variationalm}
\begin{aligned}
m(X,Y) \ &= \ \min \setb{\frac{\tr [AX]}{\tr[AY]}}{A \in \bbS_D \text{ and } \tr[AY]\neq 0} \\
&= \ \inf \setb{\frac{\tr [AX]}{\tr[AY]}}{A \in \bbS_D^\circ} \ .
\end{aligned}
\end{equation}
Furthermore, if $Z\in\bb{S}_D$, Then 
\begin{enumerate}
\item $0\le m(X,Y)\le1$ 
\item $m(X,Z)m(Z,Y)\le m(X,Y)$ 
\item $m(X,Y)m(Y,X)=1$ if and only if $X=Y$ 
\item $m(X,Y)=0$ if and only if $Y\vec{v}\neq0$ for some $\vec{v}\in\ker X$.
In particular, $m(X,Y)>0$ if $X\in  \bb{S}_D^\circ$. 
\end{enumerate}
\end{lem} 
\begin{proof} If $\lambda Y \le X$, then $\lambda \tr[AY] \le \tr[AX]$ and $\lambda \le \frac{\tr[AX]}{\tr[AY]}$ if $A\in \bbS_D$ and $\tr[AY]\neq 0$. Thus $$m(X,Y)\ \le \ \inf \setb{\frac{\tr[AX]}{\tr[AY]}}{A\in \bbS_D \text{ and } \tr[AY]\neq 0} \ . $$  
To see that the infimum is attained and is equal to $m(X,Y)$, note that if $\lambda = m(X,Y)$ then $0$ must be an eigenvalue of $X-\lambda Y$ with an eigenvector $\vec{u}$ such that $Y\vec{u}\neq 0$ (else we could increase $\lambda$ by a small amount without violating $\lambda Y \le X$).  Let $A= \vec{u}\ipc{\vec{u}}{\cdot}$.  Then $\tr[AX]=m(X,Y)\tr[AY]$ and $\tr[AY]\neq 0$.

To see that we still obtain $m(X,Y)$ if we restrict the infimum to range over $A\in \bbS_D^\circ$, note that if $\lambda Y \not \le X $, then we must have $\ipc{\vec{u}}{X\vec{u}}< \lambda \ipc{\vec{u}}{Y\vec{u}} < 0$ for some $\vec{u}$.  Since $X$ is positive, it follows that $\ipc{\vec{u}}{Y\vec{u}} \neq 0$ and $\lambda > \frac{\ipc{\vec{u}}{X\vec{u}}}{\ipc{\vec{u}}{Y\vec{u}}}$. Taking $M=\vec{u}\ipc{\vec{u}}{\cdot} + \delta \bbI$ for small enough $\delta$ we see that $\lambda > \inf\{\frac{\tr[AX]}{\tr[AY]}  :  A \in \bbS_D^\circ\}$.  Thus 
$$m(X,Y) \ = \ \inf\setb{\frac{\tr[AX]}{\tr[AY]}}{ A \in \bbS_D^\circ} \ . $$

The lower bound in part 1 is clear. To see
the upper bound note that $\frac{\tr[\bbI X]}{\tr[\bbI Y]} = 1$. For part 2, note that 
if $\lambda Z\le X$ and $\mu Y\le Z$, then $\lambda\mu Y\le X$.
For part 3, note that if $m(X,Y)m(Y,X)=1$ then $m(X,Y)=m(Y,X)=1$ so $X\le Y$
and $Y\le X$.  Finally, for part 4, note that if $Y\vec{v}\neq0$ and $X\vec{v}=0$
then $\lambda\ipc{\vec{v}}{Y\vec{v}}>0=\ipc{\vec{v}}{X\vec{v}}$ for
any $\lambda>0$. Conversely, if $Y\vec{v}=0$ for any $\vec{v}\in\ker X$,
then $Y$ is reduced by the subspace decomposition $\ker X\oplus\ran X$,
and with respect to this decomposition
\[
X=\begin{pmatrix}0 & 0\\
0 & X'
\end{pmatrix}\quad \text{and} \quad Y=\begin{pmatrix}0 & 0\\
0 & Y'
\end{pmatrix}\quad,
\]
where $X'$, $Y'$ are operators on $\ran X$. Furthermore $\ker X'=\{0\}$,
so $X'\ge\delta\mathbb{I}$ for some $\delta>0$. It follows that $\lambda Y'\le X'$
for small $\lambda >0$. Then $\lambda Y\le X$, so $m(X,Y)>0$.
\end{proof}

\begin{cor} $d_0(X,Y):= -\log m(X,Y)-\log m(Y,X)$ is a metric on $\bb{S}_D.$
\end{cor}

The metric $d_0(X,Y)$ is slightly unpleasant; it is unbounded and takes
the value $\infty$. A much nicer metric is given by
\begin{equation}\label{eq:d}
d(X,Y)=\frac{1-m(X,Y)m(Y,X)}{1+m(X,Y)m(Y,X)}\quad.
\end{equation}
\begin{lem}\label{lem:isametric} $d$ is a metric on $\bb{S}_D$
such that 
\begin{enumerate}
\item $\sup\setb{d(X,Y)}{X,Y\in\bb{S}_D}=1$, and 
\item if $X\in \bb{S}_D^\circ$ and $Y\in\bb{S}_D$, then $d(X,Y)=1$ if and only
if $Y\in\bb{S}_D\setminus  \bb{S}_D^\circ$. 
\end{enumerate}
\end{lem}

\begin{proof} Symmetry of $d$ is clear. Furthermore, $0\le d(X,Y)\le1$
and $d(X,Y)=0$ if and only if $m(X,Y)m(Y,X)=1$, which holds if and
only if $X=Y$ by Lemma \ref{lem:m}. To prove the triangle inequality, let
\[
f(s)=\frac{1-s}{1+s}=-1+\frac{2}{1+s}\quad
\]
for $0\le s\le 1$.
Then $f$ is decreasing and
\[
f(s)+f(t)\ =\ \frac{2-2st}{1+s+t+st}\ =\ 2\frac{1+st}{1+s+t+st}\frac{1-st}{1+st} = 2\frac{1}{1+\frac{s+t}{1+st}}\frac{1-st}{1+st}\quad.
\]
The maximum of $\frac{s+t}{1+st}$ over $s,t\in[0,1]$ is $1$, from
which it follows that $f(s)+f(t)\ge f(st)$. The triangle inequality
for $d$ follows from this inequality and part 2 of Lemma \ref{lem:m}.

To prove that the diameter of $\bb{S}_D$ is $1$ as claimed, we simply need to  find $X,Y\in\bb{S}_D$
with $m(X,Y)=0$. This holds, for instance, if $X\in  \bb{S}_D^\circ$ and $\ker Y\neq\{0\}$, which also leads to the result noted in item 2. \end{proof}

\begin{lem}\label{lem:dformula} Let $X,Y\in\bb{S}_D$
with $X\neq Y$. Then
\begin{equation}\label{eq:formulaford}
	d(X,Y)\ =\ \frac{\abs{u_{1}v_{2}-u_{2}v_{1}}}{u_{1}v_{2}+u_{2}v_{1}}\quad ,
\end{equation}
where $X=u_{1}A_{-}+u_{2}A_{+}$
and $Y=v_{1}A_{-}+v_{2}A_{+}$ with $A_{\pm}$  the endpoints of the intersection of $\bb{S}_D$ with the line through $X$ and $Y$.  
\end{lem} \begin{rem} 
Since $X$ and $Y$ lie on the segment connecting
$A_{\pm}$ we have $u_{1}+u_{2}=1$ and $v_{1}+v_{2}=1$. Thus we
have
\[
d(X,Y)\ =\ \frac{\abs{u_{1}-v_{1}}}{u_{1}+v_{1}-2u_{1}v_{1}}\ =\ \frac{\abs{u_{2}-v_{2}}}{u_{2}+v_{2}-2u_{2}v_{2}}\quad.
\]
\end{rem} \begin{proof} Let $t_{+}$ and $t_{-}$ be the largest and smallest
real numbers such that $tX+(1-t)Y\in\bb{S}_D$. Note that $t_{-}\le0\le1\le t_{+}$ and $A_{\pm}=t_{\pm}X+(1-t_{\pm})Y$. Furthermore
\[
u_1 = \frac{t_{+}-1}{t_{+}-t_{-}} \ , \quad u_2= \frac{1-t_{-}}{t_{+}-t_{-}}\ , \quad v_1 = \frac{t_{+}}{t_{+}-t_{-}} \ , \quad \text{and} \quad v_2 = \frac{-t_{-}}{t_{+}-t_{-}} \quad ,
\]
so equation \eqref{eq:formulaford} is equivalent to
\begin{equation}\label{eq:formulafordwitht}
	d(X,Y)\ =\  \frac{t_{+}-t_{-}}{t_{-}+t_{+}-2t_{-}t_{+}}\quad.
\end{equation}

Note that each $A_{\pm}$ must have a non-trivial
kernel. For example, if $A_{-}$ were positive-definite, then 
$A_{-}-\delta(X-Y)$ would be positive definite for small $\delta$, contradicting
the minimality of $t_{-}$. A similar argument applies to $A_+$. Furthermore we must have $\ker A_{-}\not \subset \ker A_{+}$
and $\ker A_{+}\not \subset \ker A_{-}.$ Indeed, suppose that
$\ker A_{+}\subset\ker A_{-}$. Then we would have $A_{+}-\delta A_{-}\ge0$
for small $\delta$, contradicting the maximality of $t_{+}$.  The proof that 
$\ker A_{-}\not \subset \ker A_{+}$ is similar.

Suppose that $t_{+}=1$. Then $X=A_{+}$  and $tX+(1-t)Y$ is not
positive definite whenever $t>1$, i.e., $X-\lambda Y$ is not positive
definite for any $\lambda>0$. It follows that $m(X,Y)=0$ and thus $d(X,Y)=1$, so equation \eqref{eq:formulafordwitht} holds. Similarly, if $t_{-}=0$ then $Y=A_{-}$, $d(X,Y)=1$, and equation \eqref{eq:formulafordwitht} holds.

Now suppose that $t_+>1$ and $t_-<0$. Then $X$ and
$Y$ are in the interior of the interval connecting $A_{-}$ and $A_{+}$.
Let $r=\min\setb{u_{i}/v_{i}}{i=1,2}$. Then
\[
rY=r(v_{1}A_{-}+v_{2}A_{+})\le u_{1}A_{-}+u_{2}A_{+}=X\quad.
\]
Thus $r\le m(X,Y).$ On the other hand
\[
m(X,Y)(v_{1}A_{-}+v_{2}A_{+})  =  m(X,Y)Y\le X=u_{1}A_{-}+u_{2}A_{+}\quad.
\]
Let $\vec{w}_{+}\in\ker A_{+}\setminus \ker A_{-}$. Then
$
m(X,Y)v_{1}\ipc{\vec{w}_{+}}{A_{-}\vec{w}_{+}} \le u_{1}\ipc{\vec{w}_{+}}{A_{-}\vec{w}_{+}}$.
Thus $m(X,Y)\le u_{1}/v_{1}$. Similarly, working with
$w_{-}\in\ker A_{-}\setminus A_{+}$ we find that $m(X,Y)\le u_{2}/v_{2}$
so that $m(X,Y)\le r$. Thus $m(X,Y)=r$. Similarly, $m(Y,X)=\min\setb{v_{i}/u_{i}}{i=1,2}.$
Thus,
\[
m(X,Y)m(Y,X)\ =\ \min\set{\frac{u_{1}}{v_{1}}\frac{v_{2}}{u_{2}},\frac{u_{2}}{v_{2}}\frac{v_{1}}{u_{1}}} \quad ,
\]
from which equation \eqref{eq:formulaford} follows. \end{proof}

\begin{lem} \label{lem:metric} Let
$X,Y\in\bb{S}_D$, then
$
d(X,Y)\ge\frac{1}{2}\tr\abs{X-Y}.
$
\end{lem} \begin{proof} Based on the remark following Lemma \ref{lem:dformula}, we have
\[
d(X,Y)=\frac{\abs{u_{1}-v_{1}}}{u_{1}+v_{1}-2u_{1}v_{1}}\ge|u_{1}-v_{1}|\quad,
\]
where $X=u_{1}A_{-}+u_{2}A_{+}$ and $Y=v_{1}A_{-}+v_{2}A_{+}$ with
$A_{\pm}$ as in Lemma \ref{lem:dformula}. Since $u_{2}=1-u_{1}$ and $v_{2}=1-v_{1}$, we have $X-Y=(u_{1}-v_{1})(A_{-}-A_{+})$.
Thus
\[
-\abs{u_{1}-v_{1}}(A_{-}+A_{+})\le X-Y\le\abs{u_{1}-v_{1}}(A_{-}+A_{+})\quad,
\]
so
$\tr\abs{X-Y} \le 2\abs{u_{1}-v_{1}}$.
\end{proof}

\begin{lem}\label{lem:homeo} Let $d_1(X,Y)=\tr\abs{X-Y}$ denote the trace norm metric 
on $\bb{S}_D$. Let $Y\in  \bb{S}_D^\circ$, $X\in\bb{S}_D$ and let
$X_{n}$ be a sequence in $\bb{S}_D$ such that $\lim_{n}d_{1}(X_{n},X)=0$.
Then $\lim_{n}d(X_{n},Y)=d(X,Y).$ In particular, the spaces $( \bb{S}_D^\circ,d)$
and $( \bb{S}_D^\circ,d_{1})$ are homeomorphic. \end{lem} \begin{rem*} The spaces
$(\bb{S}_D,d)$ and $(\bb{S}_D,d_{1})$ are \emph{not} homeomorphic, and look very different on the boundary $\bb{S}_D \setminus \bb{S}_D^\circ$.
For instance, if $P\in\bb{S}_D$ is an orthogonal projection onto a proper subspace,
then $Y_{t}=(1-t)P+t$ converges to $P$ in $d_{1}$ as $t\rightarrow0$,
but $d(P,Y_{t})=1$ for all $t>0$ (since $m(P,Y_{t})=0$). The space
$(\bb{S}_D,d_{1})$ is compact, but $(\bb{S}_D,d)$ has an
uncountable number of components. \end{rem*} \begin{proof} We will
show that $m(Y,X)=\lim_{n}m(Y,X_{n})$ and $m(X,Y)=\lim_{n}m(X_{n},Y)$.
Since $Y\in  \bb{S}_D^\circ$, we have $Y>\delta \bbI$ for some $\delta>0$. Given $\epsilon>0$
we have $\tr|X_{n}-X|<\epsilon$ and thus $X_{n}\le X+\epsilon \bbI$ and
$X\le X_{n}+\epsilon  \bbI$ for large enough $n$. 

We first show that $m(X,Y)=\lim_{n}m(X_{n},Y)$.  Let $\epsilon >0$. Given $\lambda\le m(X,Y)$,
so $\lambda X\le Y$, we have $\lambda Y  \le  X  \le \ X_n +\epsilon  \bbI$. Thus
$ (\lambda-\tfrac{\epsilon}{\delta})Y\le X_{n}$ for large enough $n$.  It follows that $\liminf_{n} m(X_n,Y)\ge m(X,Y)-\tfrac{\epsilon}{\delta}$.  On the other hand if $\lambda \le \limsup_n m(X_n,Y)$, then we have $\lambda Y  \le  X_{n_j}  \le  X +\epsilon  \bbI$
along a subsequence $n_j\rightarrow \infty$.  Thus $(\lambda-\tfrac{\epsilon}{\delta}) Y \le X$, and so $\limsup_{n}m(X_n,Y)\le m(X,Y)+\tfrac{\epsilon}{\delta}$. We have shown
$$ m(X,Y)-\tfrac{\epsilon}{\delta} \ \le \ \liminf_{n} m(X_n,Y) \ \le \ \limsup_{n}m(X_n,Y)\le m(X,Y)+\tfrac{\epsilon}{\delta}.$$
Taking $\epsilon \rightarrow 0$, we see that $\lim_nm(X_n,Y)=m(X,Y)$ as claimed.

Now we show that $m(Y,X)=\lim_{n}m(Y,X_{n})$.  Let $0<t<1$ and choose
$\epsilon$ small enough that $t\epsilon\le(1-t)\delta$. Given $\lambda\le m(Y,X)$, we have
\[t\lambda X_{n} \ \le \ t\lambda (X+\epsilon  \bbI ) \ \le \ t Y + (1-t)\delta  \bbI \ \le \ Y\]
for large $n$. Thus $\liminf_n m(Y,X_n)\ge t m(X,Y)$.   Similarly, given $\lambda\le\limsup_{n}m(Y,X_{n})$,
we have $ t\lambda X\le t\lambda (X_{n_j}+ \epsilon \bbI)\le Y$
along a sub-sequence $n_j\rightarrow \infty$. Thus
\[
tm(Y,X)\le\liminf_{n\rightarrow\infty}m(Y,X_{n})\le\limsup_{n\rightarrow\infty}m(Y,X_{n})\le\frac{1}{t}m(Y,X).
\]
Taking $t\rightarrow 1$, we find that $m(Y,X)=\lim_{n}m(Y,X_{n})$, completing the proof of (1).

To prove that $( \bb{S}_D^\circ,d_{1})$ and $( \bb{S}_D^\circ,d)$ are homeomorphic, note that $d_{1}\le2d$
by  Lemma \ref{lem:metric}. Thus convergence in $d$ implies convergence in $d_{1}$ on all of $\bb{S}_D$. On the other hand, if $X_{n} \rightarrow X \in  \bb{S}_D^\circ$ with respect to $d_{1}$ then  $\lim_{n}d(X_{n},X)=d(X,X)=0$, by
the first part of the lemma, so $X_{n}$ converge to $X$ with respect to $d$. \end{proof}

For any map $\phi\in \mc{P}_D$ we define the \emph{contraction coefficient}
\begin{equation}\label{eq:c(phi)}
 c(\phi)  \ = \  \sup\setb{d(\phi\cdot X,\phi\cdot Y)}{X,Y\in\bbS_D}.
 \end{equation}
 The following Lemma lists various properties of $c(\phi)$.
\begin{lem}\label{lem:contract} 
Let $\phi\in \mc{P}_D$, then
\begin{enumerate}
\item For $X,Y\in \bbS_D$, $d(\phi\cdot X,\phi\cdot Y)\le c(\phi)d(X,Y)$.
\item We have $c(\phi)\le 1$ and $c(\phi)<1$ if and only if $\phi\in \mc{P}_D^{\circ}.$ 
\item If $\phi'\in \mc{P}_D$, then $c(\phi'\circ \phi) \le c(\phi')c(\phi).$
\item $c(\phi)=c(\phi^{*})$.
\end{enumerate}
\end{lem} 
\begin{rem*} Thus, if $\phi\in \mc{P}_D^\circ$, then the projective action of $\phi$ on $\bb{S}_D$ is \emph{strictly} contractive with respect to the metric $d$. \end{rem*}

\begin{proof}
To prove (1), suppose that $\phi\in \mc{P}_D$. If $\phi\cdot X=\phi\cdot Y$, then $0=d(\phi\cdot X,\phi\cdot Y)\le c(\phi) d(X,Y).$ Now suppose that $\phi\cdot X\neq\phi\cdot Y$ and let $t_{\pm}$ and $A_{\pm}$
be as in the proof of Lemma \ref{lem:metric}. Similarly, let $A'_{\pm}=s_{\pm}\phi\cdot X+(1-s_{\pm})\phi\cdot Y$ with $s_{\pm}$ 
the largest and smallest real numbers such that $s\phi\cdot X+(1-s)\phi\cdot Y\in\bb{S}_D$. 

The linear map $\phi$ maps the two dimensional space spanned by $A_{-},A_{+}$ into the two dimensional space spanned by $A'_{-},A'_{+}$. Let the
matrix of this map (with respect to the bases $A_{-},A_{+}$ for the domain and $A'_{-},A'_{+}$ for the range) be
\[
\begin{pmatrix}\alpha & \beta\\
\gamma & \delta
\end{pmatrix}.
\]
We claim that $\alpha,\beta,\gamma,\delta\ge0$. To see that $\alpha$, $\gamma \ \ge 0$, note that
\[
\phi(A_{-})\ = \ t_{-}\phi(X)+(1-t_{-})\phi(Y)\ =\ \left[t_{-}\tr\phi(X)\phi\cdot X+(1-t_{-})\tr\phi(Y)\phi\cdot Y\right].
\]
Thus
\[
\phi\cdot A_{-}\ =\ \frac{t_{-}\tr[\phi(X)]}{\tr[\phi(A_{-})]}\phi\cdot X+\left[1-\frac{t_{-}\tr[\Phi(X)]}{\tr[\phi(A_{-})]}\right]\phi\cdot Y\quad.
\]
Since $\phi\cdot A_{-}\in\bb{S}_D$ we must have 
$s_{-}\tr[\phi(A_{-})] \le    t_{-}\tr[\phi(X)]  \le  s_{+}\tr[\phi(A_{-})]$.
Thus $\phi(A_{-})=\alpha A'_{-}+\gamma A'_{+}$ with
\[
\alpha \ =\ \frac{t_{-}\tr[\phi(X)]-s_{-}\tr[\phi(A_{-})]}{s_{+}-s_{-}}\ge0,\quad\gamma \ = \ \frac{s_{+}\tr[\phi(A_{-})]-t_{-}\tr[\phi(X)]}{s_{+}-s_{-}}\ge0\quad.
\]
The verification that $\beta\ge0$ and $\delta\ge0$ is similar.

We also have $\alpha\delta+\beta\gamma>0$. Indeed if $\alpha\delta+\beta\gamma$
were zero, then the matrix would have a zero row or a zero column.
A zero column would imply that one of $\phi(A_{-})$ or $\phi(A_{+})$ is
zero, a contradiction. A zero row would imply that $\phi(A_{+})$ and
$\phi(A_{-})$ were both proportional either to $A'_{+}$ or $A'_{-}$.
Suppose both were proportional to $A'_{-}$. Then both points would
lie on the line between $0$ and $A'_{-}$ and also on the line between
$A'_{-}$ and $A'_{+}$. Since these lines intersect only in $A'_{-}$
we would have $\phi(A_{+})=\phi(A_{-})=A'_{-}$, contradicting 
the assumption that $\phi\cdot X\neq\phi\cdot Y$.

With these preliminaries, we can now prove (1) by computing $d(\phi\cdot X,\phi\cdot Y)$. Let $X=u_{1}A_{-}+u_{2}A_{+}$,
$Y=v_{1}A_{-}+v_{2}A_{+}$. Then,
\[
\begin{aligned}d(\phi\cdot X,\phi\cdot Y)\  & =\ \frac{\abs{(\alpha u_{1}+\beta u_{2})(\gamma v_{1}+\delta v_{2})-(\gamma u_{1}+\delta u_{2})(\alpha v_{1}+\beta v_{2})}}{(\alpha u_{1}+\beta u_{2})(\gamma v_{1}+\delta v_{2})+(\gamma u_{1}+\delta u_{2})(\alpha v_{1}+\beta v_{2})}\\
 & =\ \frac{\abs{\alpha\delta-\beta\gamma}\abs{u_{1}v_{2}-u_{2}v_{1}}}{\alpha\gamma2u_{1}v_{1}+(\alpha\delta+\beta\gamma)(u_{1}v_{2}+u_{2}v_{1})+\beta\delta2u_{2}v_{2}}\\
 & \le\ \frac{\abs{\alpha\delta-\beta\gamma}}{\alpha\delta+\beta\gamma}\frac{\abs{u_{1}v_{2}-u_{2}v_{1}}}{u_{1}v_{2}+u_{2}v_{1}}\\
 & =\ d(\phi\cdot A_{-},\phi\cdot A_{+})d(X,Y)\ \le\ c(\phi)d(X,Y)\quad.
\end{aligned}
\]

Turning now to (2), if $\phi \in \mc{P}_D\setminus \mc{P}_D^\circ$, then $\phi \cdot X\in \bb{S}_D \setminus \bb{S}_D^\circ$ for some $X\in \bb{S}_D\setminus \bb{S}_D^\circ$ (otherwise $\phi$ would be strictly positive). Thus, $c(\phi)=1$, since by Lemma \ref{lem:isametric}, $d(\phi\cdot X,\phi\cdot Y)=1$ for $Y\in \bb{S}_D^\circ$.   To see that $c(\phi)<1$ for $\phi\in \mc{P}_D^\circ$, note that $\phi\cdot$ is a continuous map from $(\bb{S}_D,d_{1})$
into $( \bb{S}_D^\circ,d_{1})$.  By Lemma \ref{lem:homeo},
$
F(X,Y)=d(\phi\cdot X,\phi\cdot Y)
$
is a continuous map of $\bb{S}_D\times\bb{S}_D$ into $\R$,
where we take the $d_{1}$-product topology on $\bb{S}_D\times\bb{S}_D$.
Since $\bb{S}_D\times\bb{S}_D$ is compact we conclude that
there are $X,Y\in\bb{S}_D$ such that $
c(\phi)=d(\phi\cdot X,\phi\cdot Y)$. Since $\phi\cdot X$, $\phi\cdot Y\in  \bb{S}_D^\circ$ we have
$
0<m(\phi\cdot X,\phi\cdot Y)<1$ and $0<m(\phi\cdot Y,\phi\cdot X)$,
so that $c(\phi)=d(\phi\cdot X,\phi\cdot Y)<1$.  

To prove (3), note that $\phi'\circ \phi \cdot X = \phi'\cdot (\phi \cdot X)$, so that
$$d(\phi'\circ \phi \cdot X, \phi'\circ \phi \cdot Y) \ \le \ c(\phi') d(\phi \cdot X, \phi \cdot Y) \ \le \ c(\phi)c(\phi') d(X,Y) \ , $$
by part (1).

Finally, to prove that $c(\phi)=c(\phi^{*})$, we use the variational formula \eqref{eq:variationalm} which implies 
\[
m(X,Y)m(Y,X)=\inf\setb{\frac{\tr [AX]}{\tr [A'X]}\frac{\tr [A'Y]}{\tr [AY]}}{A,A'\in  \bb{S}_D^\circ}\quad.
\]
It follows that
\begin{align*}
m(\phi\cdot X,\phi\cdot Y)m(\phi\cdot Y,\phi\cdot X) 
 & =\inf\setb{\frac{\tr[\phi^{*}(A)X]}{\tr[\phi^{*}(A')X]}\frac{\tr[\phi^{*}(A')Y]}{\tr[\phi^{*}(A)Y]}}{A,A'\in  \bb{S}_D^\circ}\\
 & \ge\ \inf\setb{m(\phi^{*}\cdot A,\phi^{*}\cdot A')m(\phi^{*}\cdot A',\phi^{*}\cdot A)}{A,A'\in  \bb{S}_D^\circ}\quad,
\end{align*}
and thus that
\begin{multline*}
\inf\setb{m(\phi^{*}\cdot X,\phi^{*}\cdot Y)m(\phi^{*}\cdot Y,\phi^{*}\cdot X)}{X,Y\in  \bb{S}_D^\circ}\\
=\inf\setb{m(\phi\cdot X,\phi\cdot Y)m(\phi\cdot Y,\phi\cdot X)}{X,Y\in  \bb{S}_D^\circ}\quad,
\end{multline*}
from which it follows that $c(\phi)=c(\phi^{*}).$ \end{proof} 

\subsection{Existence of $Z_{0}$ and $Z_0'$}\label{sec:existence}
We start by proving Lemma \ref{lem:Assequiv}, which states the equivalence of Assumption \ref{ass:main} to two conditions, which we reformulate here in the notation of the \S\ref{sec:notation}:
\begin{enumerate}
    \item  For some $n_{0}>0$, $ 	\Pr\left[\Phi_{n_0}\in \mc{P}_D^\circ \right]\ > \ 0 $.
    \item  With probability one,  $\phi_0 \in \mc{P}_D$.
\end{enumerate}
Recall that $\Phi_{N}=\phi_{N}\circ\cdots\circ \phi_{0}$, where $\phi_{n}=\phi_{T^{n}\omega}$.
Let
$$ \tau \ = \ \inf \{N_0 \ge 0 \ : \ \Phi_N \in \mc{P}_D^\circ \text{ for } N \ge N_0\} \ . $$
Note that Assumption \ref{ass:main} is equivalent to the statement that $\tau <\infty$ with probability one.

\begin{proof}[Proof of Lemma \ref{lem:Assequiv}] We first show that Conditions (1) and (2) imply $\tau <\infty$ with probability one.  
By ergodicity and condition (1), 
\[
\Pr\left[\bigcup_{k\ge0}\set{\Phi_{n_{0};T^{k}\omega}\in \mc{P}_D^\circ}\right]=1\quad.
\]
Thus with probability $1$ there is $\sigma<\infty$ such that $
\Phi_{n_{0};T^{\sigma}\omega}=\phi_{\sigma+n_0}\circ\cdots\circ\phi_{\sigma}\in \mc{P}_D^\circ$.  By Condition (2) and the shift invariance of probabilities, we have $\phi_n\in \mc{P}_D$ for all $n$, with probability one. By Corollary \ref{cor:ideal}, it follows that
$$\Phi_N \ = \ \phi_N \circ \cdots \phi_{n_0+\sigma+1}\circ \Phi_{n_0;T^{\sigma}\omega} \circ \phi_{\sigma-1} \circ \cdots \circ \phi_0 $$
is strictly positive for $N\ge n_{0}+\sigma$, so $\tau\ge n_0 + \sigma$

Conversely, note that Assumption \ref{ass:main} implies Condition (1) directly. To prove Condition (2), note that for $N>0$, we have $\ker \Phi_N \supset \ker \phi_0$.  It follows that $\ker\phi_0\cap \bbP_D=\{0\}$ if $\Phi_N\in \mc{P}_D^\circ$ for some $N$.  Thus Assumption \ref{ass:main} implies that $\ker \phi_0 \cap \bbP_D = \{0\}$ with probability one. To prove the corresponding statement for $\phi_0^*$, first note that if $\Phi_N=\phi_N\circ \Phi_{N-1}$ is strictly positive, then we have  $\tr \phi_N^*(A) \Phi_{N-1}(B) > 0 $
for every $A,B\in \bbP_D\setminus\{0\}$.  Thus Assumption 1 implies that, with probability one, $\ker \phi_N^* \cap \bbP_D = \{0\}$ for all sufficiently large $N$.  Let $A_{M}$ denotes the event $$A_{M}\ = \ \bigcap_{N\ge M} \setb{\omega }{ \ker \phi_{N;\omega}^* \cap \bbP_D = \{0\} } \ . $$ 
Then $(A_M)_{M=0}^\infty$ is an increasing sequence and  $\Pr[\bigcup_{M} A_{M}]=1$.  Thus $\lim_M \Pr[A_M]=1$.  However, $A_M = T^M(A_0)$ so $\Pr[A_M]=\Pr[A_0]$ for all $M$. We conclude that $\Pr[A_0]=1$ and thus that $\ker \phi_0^* \cap \bbP_D = \{0\}$ with probability one.
\end{proof}

\begin{lem}\label{lem:Phicontract}
Let $c_{N}=c(\Phi_{N})$, with $c(\cdot)$ the contraction coefficient
in \eqref{eq:c(phi)}. Then
\[
\lim_{N\rightarrow\infty}c_{N}^{\nicefrac{1}{N}}=\inf_{N}c_{N}^{\nicefrac{1}{N}}\equiv\kappa
\]
exists almost surely, where $\kappa\in[0,1)$ is non-random and
\begin{equation}\label{eq:kappaformula}
\ln\kappa=\lim_{N\rightarrow\infty}\frac{1}{N}\Ev{\ln c_{N}}=\inf_{N}\frac{1}{N}\Ev{\ln c_{N}}\quad.
\end{equation}
\end{lem} 

\begin{proof}Note that
$ \ln c_{N+M} = \ln c(\Phi_{N+M}) = \ln c(\Phi_{N;T^{M+1}\omega}\circ\Phi_{M}). \le \ln c(\Phi_{N;T^{M+1}\omega})+\ln c(\Phi_{M})$ 
and $\ln c(\Phi_{N})\le0$ it follows from the subadditive ergodic
theorem \cite{Kingman} that the limit and infimum exist, and that   equation \eqref{eq:kappaformula} holds. 
Since $0\le c(\Phi_{N})\le1$, we have $0\le\kappa\le1$. By Condition (1) of Lemma \ref{lem:Assequiv} and
Lemma \ref{lem:contract}, we see that $c_{n_0}<1$
with positive probability. Thus $n_0^{-1}\Ev{\ln c_{n_0}}<0$ and so $\ln\kappa<0$.
\end{proof}

We can now prove the existence of the limiting processes: \begin{lem}\label{lem:Z'}Let
$\Phi_{N}^*$ and $L_{N}$ be as in   equation \eqref{eq:PhiStar_N}. As $N\rightarrow \infty$, $L_{N}$ converges
almost surely to a limit $Z_{0}'$ such that:
\begin{enumerate}
\item $Z_{0}'\in  \bb{S}_D^\circ$ almost surely;  
\item $\phi_{0}^*\cdot Z_{0;T\omega}'=Z_{0;\omega}'$; and
\item for $Y\in\bb{S}_D$ and $N\ge0$, we have $d(\Phi_{N}^*\cdot Y,Z_{0}')\ \le c(\Phi_{N}).$ 
\end{enumerate}
\end{lem} \begin{proof} Let $B_{N}=\Phi_{N}^*\cdot\bb{S}_D$, so $B_N\subset B_{N-1}$.
It follows from Assumption \ref{ass:main} that $B_{N}\subset  \bb{S}_D^\circ$
for large enough $N$.  Thus, by Lemma \ref{lem:homeo},
$B_{N}$ is compact in the $d$-topology for large $N$ (since $B_N$ is compact in the $d_1$-topology for every $N$). Thus $\cap_{N}B_{N}$ is non-empty.
On the other hand,
\[
\diam B_{N}\ =\ \sup\setb{d(\Phi_{N}^*\cdot X,\Phi_{N}^*\cdot Y)}{X,Y\in\bb{S}_D}\ \le\ c(\Phi_{N}) \ \rightarrow \ 0 \quad,
\]
by Lemmas \ref{lem:contract} and \ref{lem:Phicontract}. Thus $\cap_{N}B_{N}=\{Z_{0}'\}$
for a single point $Z_{0}'$. It is clear that $Z_{0}'\in  \bb{S}_D^\circ$ almost
surely.

We claim that $L_{N}\in B_{N}$. Indeed, since $\Phi_{N}^*(L_{N})=\lambda_{N}L_{N}$
and $\tr [L_{N}]=1$, it follows that $\Phi_{N}^*\cdot L_{N}=L_{N}$.
Thus,  $d(L_{N},Z_{0}') \le \diam B_N \rightarrow0$ almost surely. It follows that
$
Z_{0;T\omega}'=\lim_{N\rightarrow\infty}L_{N;T\omega}$.
However $L_{N;T\omega}$ is a normalized  eigenmatrix for
$
\Phi_{N;T\omega}^*=\phi_{1}^*\circ\cdots\circ\phi_{N+1}^*$.
Thus $\phi_{0}^*\cdot L_{N;T\omega}=\ \Phi_{N+1}^*\cdot L_{N;T\omega}\in B_{N+1},$
from which it follows that $\phi_{0}^*\cdot L_{N;T\omega} \rightarrow Z_{0}'$.
Finally, let $Y\in\bb{S}_D$. Then $\Phi_{N}^*\cdot Y\in B_{N}$,
so $d(\Phi_{N}^*\cdot Y,Z_{0}')\ \le\ \diam B_{N}\ \le\ c(\Phi_{N})$
as claimed. \end{proof}

A similar argument can be applied to $\Phi_{-N}^*$ 
to conclude the existence and properties of $Z_{0}$.  To this end, let $\psi_n = \phi_{-n}^*$ and 
$$\Psi_N \ = \ \psi_N \circ \cdots \circ \psi_0 \ = \ \phi_{-N}^* \circ \cdots \circ \phi_0^* \ = \  \Phi_{-N}^* \ . $$   First we note that the process $\Psi_N$ satisfies Assumption \ref{ass:main}:
\begin{lem}\label{lem:stopping} With probability one, there is $N_0'<\infty $ such that $\Psi_N \in \mc{P}_D^\circ$ for all $N\ge N_0'$.
\end{lem}
\begin{proof} We will show that Conditions (1) and (2) of Lemma \ref{lem:Assequiv} hold. Condition (2) for $\psi_0$ follows from the corresponding statement for $\phi_0$, since $\psi_0=\phi_0^*$ and $\psi_0^* =\phi_0$. To see that Condition (1) holds, note that $\Psi_{n_0}^*=\Phi_{-n_0;\omega}=\Phi_{n_0;T^{-n_0}\omega}$.  Thus by the shift invariance of probabilities, $\Pr[\Psi_{n_0}^* \in \mc{P}_D^\circ]  =  \Pr[\Psi_{n_0} \in \mc{P}_D^\circ] > 0 $.
\end{proof}

The existence of $Z_0$ follows directly from Lemma \ref{lem:Z'} applied to $\Psi_N = \Phi_{-N}^*$:
\begin{lem}\label{lem:Z}Let
$\Phi_{N}$ and $R_{N}$ be as in   equation \eqref{eq:Phi_N}. As $N\rightarrow -\infty$, $R_{N}$ converges
almost surely to a limit $Z_{0}$ such that:
\begin{enumerate}
\item $Z_{0}\in  \bb{S}_D^\circ$ almost surely; 
\item $\phi_{0}\cdot Z_{0;T^{-1}\omega}=Z_{0}$; and
\item for $Y\in\bb{S}_D$ and $N\ge0$, we have $d(\Phi_{-N} \cdot Y,Z_{0})\ \le c(\Phi_{-N})$.
\end{enumerate}
\end{lem}

\section{Proofs of the Theorems}\label{sec:proofs}
\subsection{Proof of Theorem \ref{thm:main}} We have already
shown the existence of the limits $\lim_{N\rightarrow \infty} L_N=Z_{0}'$ and $\lim_{N\rightarrow -\infty }R_N=Z_{0}$. Let $Z_{n}=Z_{0;T^{n}\omega}$ and $Z_{n}'=Z_{0;T^{n}\omega}'$.
Then $Z_{0}=\phi_{0}\cdot Z_{-1}$ and $Z_{0}=\phi_{0}^{*}\cdot Z_{1}'$
by Lemmas \ref{lem:Z'} and \ref{lem:Z}. Thus 
$Z_{n}=\phi_{n}\cdot Z_{n-1}$ and $Z_n'=\phi_{n}^{*}\cdot Z_{n+1}'$
as claimed. \qed

\subsection{Proof of Theorem \ref{thm:bound}}

 Let $\mu\in(\kappa,1)$
be as in Lemma \ref{lem:Phicontract}. Given $m<n$, let
$
\Psi_{n,m}=\phi_{n}\circ\cdots\circ\phi_{m}$
and $P_{n,m}(M)=\tr[ Z_{m}'M]\; Z_{n}.$  To prove   equation \eqref{eq:main}, we must show that
\begin{equation}\label{eq:mainagain}
\tr \left | \frac{1}{\tr[\Psi_{n,m}^*(\bbI)]}\Psi_{n,m}(M)-P_{n,m}(M) \right | \ \le \ C_{\mu,x} \mu^{n-m} \tr |M|
\end{equation}
whenever $m\le x \le n$. In fact, it suffices to prove equation \eqref{eq:mainagain} for $M\in\bb{S}_D$. Indeed, any matrix $M$ can be written as a linear combination 
\begin{equation}\label{eq:Msplit}
	M=\sum_{j=1}^4 a_j M_j \qquad  \text{with }M_j\in \bb{S}_D\text{ and } \sum_{j=1}^4|a_j|\le 2\tr[|M|]\quad .
\end{equation} 
Thus equation \eqref{eq:mainagain} for $M\in \bb{S}_D$ implies the same bound for general $M$, with the constant $C_{\mu,x}$ increased by a factor of $2$. (To see that   equation \eqref{eq:Msplit} holds, note that for self-adjoint $M$ we have $M= \tr[M_+] \rho_{+}-\tr[M_-]\rho_{-}$, where $\rho_\pm = \frac{1}{\tr[M_\pm]} M_\pm$ with $M_\pm$ the positive and negative parts of $M$. 
For a general matrix $M$, we proceed by applying this decomposition to the real and imaginary parts 
$M=M_{r}+iM_{i}$, where $M_{r}=\frac{1}{2}(M+M^{\dagger})$  and $ M_{i}=\frac{1}{2i}(M-M^{\dagger})$.)

Now let $M\in\bb{S}_D$ be fixed. Note that
$\Psi_{n,m} = \Phi_{m-n;T^{n}\omega}$ and $\Psi_{n,m}^{*} = \Phi_{n-m;T^{m}\omega}^{*}$.
By Lemma \ref{lem:metric} and Lemma \ref{lem:Z} we
have
\[
\tr\left[\abs{\frac{1}{\tr\Psi_{n,m}(M)}\Psi_{n,m}(M)-Z_{n}}\right]\ \le\ 2c(\Psi_{n,m})\quad.
\]
By Lemma \ref{lem:metric} and Lemma \ref{lem:Z'}, we have
\[
\abs{\frac{\tr[\Psi_{n,m}(M)]}{\tr[\Psi_{n,m}^{*}(\bbI)]}-\tr [Z_{m}'M]}\ \le\ 2c(\Psi_{n,m})\quad,
\]
where we have noted that $ \tr[\Psi_{n,n}^{*}(\bbI)M] =\tr [\Psi_{n,m}(M)]$.
Thus
\begin{equation}
\tr\left[\abs{\frac{1}{\tr[\Psi_{n,m}^{*}(\bbI)]}\Psi_{n,m}(M)-P_{n,m}(M)}\right]\ \le\ 2c(\Psi_{n,m})\left(1+\frac{\tr[\Psi_{n,m}(M)]}{\tr[\Psi_{n,m}^{*}(\bbI)]}\right)\ \le4c(\Psi_{n,m}),\label{eq:key}
\end{equation}
since $\tr[\Psi_{n,m}(M)]=\tr[\Psi_{n,m}^{*}(\bbI)M]\ \le\ \tr[\Psi_{n,m}^{*}(\bbI)]$
for $M\in  \bb{S}_D^\circ$.

To prove   equation \eqref{eq:mainagain}, first suppose that $m\le x<n$. Then
$c(\Psi_{n,m}) \le c(\Psi_{n,x+1})c(\Psi_{x,m})$.
By Lemma \ref{lem:Phicontract} we have
$c(\Psi_{n,x+1}) \le D_{\mu,x}\mu^{n-x}$ and $c(\Psi_{x,m}) \le D_{\mu,x}\mu^{x-m}$
for suitable $D_{\mu,x}<\infty$. Thus
$c(\Psi_{n,m}) \le D_{\mu,x}^{2}\mu^{n-m}$,
so equation \eqref{eq:mainagain} follows from equation \eqref{eq:key}. For
$m\le x=n$, we have
$c(\Psi_{n,m}) =  c(\Psi_{x,m})  \le D_{\mu,x}\mu^{x-m},$
so   equation \eqref{eq:mainagain} holds in this case as well. \qed

\subsection{Proof of Theorem \ref{thm:corelations}} We have
\begin{multline*}
W(O_{2}O_{1})-W(O_{2})W(O_{1})\\
=\tr\left[\left(\wt{O}_{2}-W(O_2)\wt{\Psi}_{2}\right)\circ\wt{\phi}_{m_2-1}\circ\cdots\circ\wt{\phi}_{n_{1}+1}\circ\left(\wt{O}_{1}(\wt{Z}_{m_1-1})-W(O_1)\wt{\Psi}_{1}(\wt{Z}_{m_1-1})\right)\right].
\end{multline*}
By Theorem \ref{thm:bound}, there is $0<\mu<1$ such that
\[
\tr\left[\abs{\wt{\phi}_{m_2-1}\circ\cdots\circ\wt{\phi}_{n_{1}+1}(M)-\wt{Z}_{n_1}\tr M}\right] \le\ C_{\mu,x}\mu^{m_{2}-n_{1}}\tr[|M|]\quad,
\]
for any $D\times D$ matrix $M$, where we have used the fact that
$\wt{\phi}_{n_1+1}^{*}\circ\cdots\wt{\phi}_{m_2-1}^{*}(\bbI)=\bbI$.
Equation \eqref{eq:corelations} then follows. \qed \\

\subsection*{Acknowledgments} RM acknowledges the support of the
IBM Research Frontiers Institute and funding from the MIT-IBM Watson AI Lab under the project Machine Learning in Hilbert space. JS acknowledges the support of 
the National Science Foundation under Grant No. 1500386 and Grant No. 1900015, and thanks Lubashan Pathirana for insightful discussions. We thank Natalie Taylor and Jules Murphy at IBM for their help with the graphic design of Figure 1.

\bibliographystyle{plain}
\bibliography{mybib}

\end{document}